\documentclass[aps,prl,groupedaddress,superscriptaddress,twocolumn,notitlepage,bibnotes]{revtex4-1}

\usepackage{amsthm,amsmath,amssymb,amsbsy}

 \usepackage{hyperref}
\usepackage{url}

\usepackage{graphicx}
\usepackage{bm,bbm}%
\usepackage{braket}
\usepackage{enumerate}
\usepackage{booktabs,multirow}
\usepackage{mathtools,bbding}
\usepackage{microtype}
\usepackage{mathrsfs}

\usepackage{txfonts,newtxmath}

\usepackage[most,breakable]{tcolorbox}

\newcommand{\nc}{\newcommand}
\newcommand{\rnc}{\renewcommand}

\makeatletter
\newcommand*\rel@kern[1]{\kern#1\dimexpr\macc@kerna}
\newcommand*\widebar[1]{%
  \begingroup
  \def\mathaccent##1##2{%
    \rel@kern{0.8}%
    \overline{\rel@kern{-0.8}\macc@nucleus\rel@kern{0.2}}%
    \rel@kern{-0.2}%
  }%
  \macc@depth\@ne
  \let\math@bgroup\@empty \let\math@egroup\macc@set@skewchar
  \mathsurround\z@ \frozen@everymath{\mathgroup\macc@group\relax}%
  \macc@set@skewchar\relax
  \let\mathaccentV\macc@nested@a
  \macc@nested@a\relax111{#1}%
  \endgroup
}

\rnc{\thesection}{\arabic{section}}
\rnc{\thesubsection}{\thesection.\arabic{subsection}}
\rnc{\thesubsubsection}{\thesubsection.\arabic{subsubsection}}

\usepackage{etoolbox}
\hypersetup{linktoc=all}

\newtheorem{definition}{Definition}
\newtheorem{proposition}[definition]{Proposition}
\newtheorem{lemma}[definition]{Lemma}

\newtheorem{theorem}[definition]{Theorem}
\newtheorem{corollary}[definition]{Corollary}

\theoremstyle{definition}
\newtheorem*{rem}{Remark}

\DeclareMathOperator{\Tr}{Tr}

\DeclareMathOperator{\CN}{CN}

\DeclareMathOperator{\rank}{rank}
\DeclareMathOperator{\sspan}{span}

\nc{\tcr}[1]{{\color{red} #1}}

\nc{\psic}{\psi^{c}}
\nc{\two}[1]{\underline{2^{d-#1}}}
\rnc{\H}{\mathcal{H}}
\nc{\Hanc}{\mathcal{H}_{\text{anc}}}
\nc{\psianc}{\psi_{\text{anc}}}
\nc{\lampow}{\lambda^{1/d}}

\nc{\norm}[2]{\left\lVert#1\right\rVert_{\,#2}}
\nc{\proj}[1]{\ket{#1}\!\!\bra{#1}}
\nc{\pro}[1]{#1 #1^\dagger}
\nc{\lnorm}[2]{\left\lVert#1\right\rVert_{\ell_{#2}}}

\nc{\RR}{{{\mathbb R}}}
\nc{\CC}{{{\mathbb C}}}
\nc{\FF}{{{\mathbb F}}}
\nc{\NN}{{{\mathbb N}}}
\nc{\ZZ}{{{\mathbb Z}}}

\nc{\MIO}{{\text{\rm MIO}}}
\nc{\DIO}{{\text{\rm DIO}}}
\nc{\IO}{{\text{\rm IO}}}
\nc{\SIO}{{\text{\rm SIO}}}
\nc{\PIO}{{\text{\rm PIO}}}

\nc{\SEP}{{\text{SEP}}}
\nc{\NS}{{\text{NS}}}
\nc{\LOCC}{{\text{LOCC}}}
\nc{\PPT}{{\text{PPT}}}
\nc{\EXT}{{\text{EXT}}}
\nc{\OLOCC}{{\text{1-LOCC}}}
\nc{\SEPP}{{\text{SEPP}}}

\nc{\MC}{{\text{\rm MC}}}

\nc{\cE}{\mathscr{E}}

\rnc{\*}{\textup{*}}
\nc{\<}{\left\langle}
\rnc{\>}{\right\rangle}

\rnc{\bar}{\;\rule{0pt}{9.5pt}\right|\;}

\nc{\lset}{\left\{\left.}
\nc{\rset}{\right\}}
\nc{\lsetr}{\left\{\,}
\nc{\rsetr}{\right.\right\}}
\nc{\barr}{\,\rule{0pt}{9.5pt}\left|\;}
\nc{\ketbra}[2]{\ket{#1}\!\!\bra{#2}}

\newcommand{\ketbraa}[2]{\ket{#1}\!\!\bra{#2}}

\nc{\logfloor}[1]{\left\lfloor {#1} \right\rfloor_{\log}}

\newcommand{\texteq}[1]{\stackrel{\mathclap{\mbox{\text{\scriptsize #1}}}}{=}}
\newcommand{\textleq}[1]{\stackrel{\mathclap{\mbox{\text{\scriptsize #1}}}}{\leq}}

\newcommand{\id}{\mathbbm{1}}

\newcommand{\N}{\mathcal{N}}

\renewcommand{\O}{\mathcal{O}}
\newcommand{\I}{\mathcal{I}}
\nc{\MM}{\widetilde{\M}}
\nc{\Ml}{\M^{\leq}}

\nc{\mleq}{\preceq}
\nc{\mgeq}{\succeq}

\nc{\ox}{\otimes}

\nc{\wt}{\widetilde}

\nc{\SDP}{\text{\rm SDP}}

\nc{\cc}{{\circ\circ}}
\nc{\mnorm}[1]{\norm{#1}{[m]}}

\nc{\F}{\mathcal{F}}
\nc{\M}{\mathcal{M}}

\let\oldproofname\proofname
\rnc{\proofname}{\rm\bf{\oldproofname}}
\rnc{\qedsymbol}{{\color{gray!50!black}\rule{0.6em}{0.6em}}}

\newcommand{\bb}{\begin{equation}}
\newcommand{\bbb}{\begin{equation*}}
\newcommand{\ee}{\end{equation}}
\newcommand{\eee}{\end{equation*}}

\nc{\note}[1]{{\color{blue!90!black} #1}}

\newcommand{\notts}{\affiliation{School of Mathematical Sciences and Centre for the Mathematics and Theoretical Physics of Quantum Non-Equilibrium Systems, University of Nottingham, University Park, Nottingham NG7 2RD, United Kingdom}}

\begin{document}
\title{Generic bound coherence under strictly incoherent operations}

\author{Ludovico Lami}
\email{ludovico.lami@gmail.com}
\notts

\author{Bartosz Regula}
\email{bartosz.regula@gmail.com}
\notts

\author{Gerardo Adesso}
\email{gerardo.adesso@nottingham.ac.uk}
\notts

\begin{abstract}
We compute analytically the maximal rates of distillation of quantum coherence under strictly incoherent operations (SIO) and physically incoherent operations (PIO), showing that they coincide for all states, and providing a complete description of the phenomenon of bound coherence. In particular, we establish a simple, analytically computable necessary and sufficient criterion for the asymptotic distillability under SIO and PIO. We use this result to show that almost every quantum state is undistillable --- only pure states as well as states whose density matrix contains a rank-one submatrix allow for coherence distillation under SIO or PIO, while every other quantum state exhibits bound coherence. This demonstrates fundamental operational limitations of SIO and PIO in the resource theory of quantum coherence. We show that the fidelity of distillation of a single bit of coherence under SIO can be efficiently computed as a semidefinite program, and investigate the generalization of this result to provide an understanding of asymptotically achievable distillation fidelity.
\end{abstract}

\maketitle


\textbf{\em Introduction}.---%
The resource theory of quantum coherence~\cite{aberg_2006,baumgratz_2014,winter_2016,streltsov_2017} has found extensive use in the characterization of a signature intrinsic feature of quantum mechanics --- superposition --- and our ability to manipulate it efficiently within a resource-theoretic framework~\cite{horodecki_2012,delrio_2015,coecke_2016,chitambar_2018-1}. Typically, the properties of a resource are investigated under a suitable set of allowed free operations, reflecting the constraints placed on the manipulation of the given resource~\cite{horodecki_2012,chitambar_2018-1}. In spite of the fact that the resource theory of coherence has found use in a variety of practical settings~\cite{streltsov_2017}, no physically compelling set of assumptions has yet emerged which could single out a unique class of free operations under which the operational features of coherence should be investigated, mirroring the fundamental role of local operations and classical communication in the resource theory of entanglement~\cite{horodecki_2009}.
This has motivated the definition and characterization of a multitude of possible sets of free operations, and sparked efforts to compare their operational power~\cite{du_2015, yuan_2015, yadin_2016, chitambar_2016, marvian_2016, chitambar_2016-1, vicente_2017, liu_2017, streltsov_2017-2, zhao_2018, regula_2017, chitambar_2018, egloff_2018, fang_2018, theurer_2018, regula_2018-1, zhao_2018-1}.
However, many definitions of free operations stemming from meaningful physical considerations, such as  physically incoherent operations (PIO)~\cite{chitambar_2016},  translationally-covariant incoherent operations~\cite{marvian_2016}, or  genuinely incoherent operations~\cite{vicente_2017}, were found to be too limited in their operational capabilities, suggesting that any useful resource theory of coherence would require a larger set of maps. On the other hand, strictly larger sets of maps such as  maximally incoherent operations (MIO)~\cite{aberg_2006},  incoherent operations (IO)~\cite{baumgratz_2014}, or  dephasing-covariant incoherent operations (DIO)~\cite{chitambar_2016,marvian_2016}, while operationally powerful, might be considered as too permissive and lacking a physically implementable form.


The class of \textit{strictly incoherent operations} (SIO)~\cite{winter_2016,yadin_2016} appeared to be a promising candidate for a natural class of operations satisfying desirable resource-theoretic criteria while at the same time being motivated on physical grounds and experimentally implementable, causing it to find widespread use in the resource theory of coherence~\cite{winter_2016, yadin_2016, chitambar_2016-1, streltsov_2017-2, biswas_2017}. SIO is an easy to characterize and seemingly powerful choice of free operations, allowing in particular for a ``golden unit'' of coherence represented by the maximally coherent state $\ket{\Psi_m}$, which can be transformed into any other state using SIO~\cite{baumgratz_2014}. Although strictly smaller than the sets IO and DIO, SIO includes many more transformations than PIO, and its operational capabilities did not appear to be too limited --- for instance, SIO have exactly the same power as IO as far as pure-to-pure state transformations are concerned~\cite{du_2015-1,chitambar_2016-1}, as well as in the context of coherence dilution~\cite{winter_2016,zhao_2018}; they match the power of DIO in probabilistic distillation from pure states~\cite{fang_2018}; and even the largest class of free operations, MIO, cannot perform better than SIO in one-shot distillation from pure states~\cite{regula_2017}, in assisted coherence distillation~\cite{regula_2018-1}, and in all single-qubit state transformations~\cite{chitambar_2016-1}. On the other hand, there do exist tasks in which the limitations of SIO become apparent --- in particular, unlike the larger sets IO, DIO, and MIO, the class SIO has recently been found to exhibit {\em bound coherence}~\cite{zhao_2018-1}, i.e., there are coherent states from which no coherence can be distilled by such operations. While the same phenomenon was known for PIO, it was arguably unexpected for SIO. It is, however, not known how common this property is among all quantum states, nor to what extent it limits the operational power of SIO beyond specific examples --- as is known from entanglement theory~\cite{horodecki_2009}, the mere existence of undistillable states does not inhibit a class of operations from being useful in manipulating a resource.
Indeed, a complete description of coherence distillation under SIO has been a long-standing open problem in the resource theory of coherence \cite{winter_2016,yadin_2016,bendana_2017,streltsov_2017,zhao_2018-1}, and its solution would shed light on concretely achievable possibilities in coherence manipulation.

In this work, we solve this problem completely: namely, we analytically compute the maximal rates of coherence distillation under SIO and PIO, showing that they coincide on all states. By introducing an SIO coherence monotone which does not change when multiple copies of a state are considered, we establish a simple criterion to decide whether a given quantum system is asymptotically distillable or not. In the former case, we derive an upper bound on the SIO distillable coherence and show that it can in fact be achieved by an explicitly constructed PIO protocol. This leads us to the surprising conclusion that the distillable coherence is the same under SIO and PIO. Moreover, it also shows that the optimal SIO distillation protocol can be realized by appending incoherent ancillae, applying incoherent unitaries, and making incoherent measurements; thus, it is easily implementable in practice. Our findings establish in particular that bound coherence is a {\em generic} phenomenon. Specifically, we show that almost all quantum states are undistillable under SIO, with the only distillable ones being those whose density matrix contains a submatrix proportional to a pure state. This demonstrates fundamental limitations of SIO in the resource theory of coherence. To arrive at the above results we introduce a plethora of tools of independent interest, including an efficiently computable semidefinite programming (SDP) expression characterizing the maximal achievable fidelity in the distillation of a single bit of coherence, and an entire new family of SIO monotones. Our work substantially advances the theoretical and practical study of quantum coherence.

\textbf{\em A new SIO monotone}. --- %
Let us begin by recalling the basic formalism of the resource theory of quantum coherence. The set $\I$ of free states, known as  incoherent states, consists of all density matrices diagonal in a given $d$-dimensional orthonormal basis $\{\ket{i}\}$. We will denote by $\Delta$ the dephasing map, defined by $\Delta(\cdot) = \sum_i \proj{i} (\cdot) \proj{i}$, and by $\ket{\Psi_m} = \frac{1}{\sqrt{m}} \sum_{i=1}^m \ket{i}$ the maximally coherent state of dimension $m$. As for the free operations, we will focus on strictly incoherent operations (SIO), defined as those channels $\Lambda$ that admit a Kraus decomposition $\Lambda(\cdot) = \sum_\alpha K_\alpha (\cdot) K_\alpha^\dagger$ such that $K_\alpha \Delta(\rho) K_\alpha^\dagger = \Delta(K_\alpha \rho K_\alpha^\dagger)$ for all $\alpha$ and $\rho$.
We will also consider the subset of physically incoherent operations (PIO), which are all maps that admit an incoherent dilation (i.e.~can be implemented by appending an incoherent ancilla and performing incoherent unitaries, permutations, and incoherent projections only) \cite{chitambar_2016}.

We now introduce a straightforwardly computable quantity that we name \textit{maximal coherence} of $\rho$, defined by
\vspace{-.5\baselineskip}\begin{equation}\begin{aligned}\label{eta}
	\eta(\rho) \coloneqq \max_{i \neq j} \frac{|\rho_{ij}|}{\sqrt{\rho_{ii}\rho_{jj}}}\, ,
\vspace{-.5\baselineskip}\end{aligned}\end{equation}
where the optimization is over all choices of indices such that $\rho_{ii} \neq 0 \neq \rho_{jj}$, with $\eta(\rho) = 0$ if no such choice exists. Alternatively, this quantity can be understood as the largest modulus of an off-diagonal element of the matrix $\Delta(\rho)^{-1/2}\rho\Delta(\rho)^{-1/2}$.

We first notice that for all states $\rho$, one has $0\leq \eta(\rho)\leq 1$. This follows from the positivity of the principal minor of $\rho$ of order $2$ corresponding to the rows and columns identified by indices $i$ and $j$, which implies that $\rho_{ii}\rho_{jj}\geq |\rho_{ij}|^2$ for any choice of $i,j$. Moreover, we see by definition that $\eta(\rho)=0$ iff $\rho$ is incoherent, and $\eta(\rho)=1$ iff there are indices $1\leq i\neq j\leq d$ such that $\Pi_{ij}\rho\Pi_{ij}$ is proportional to a pure state, where $\Pi_{ij}$ is the projector onto $\sspan\{\ket{i},\ket{j}\}$.

An important property of $\eta$ is that it is, in fact, monotonically non-increasing under SIO. Precisely, consider an SIO operation $\Lambda$ acting on a $d$-dimensional system, which can be written as
\begin{equation}\begin{aligned}
	\Lambda(\cdot) = \sum\nolimits_\alpha U_{\pi_\alpha}^\intercal D_\alpha (\cdot) D_\alpha^* U_{\pi_\alpha}\, , \label{SIO Kraus}
\end{aligned}\end{equation}
where the $\pi_\alpha$ are permutations, $U_{\pi_\alpha} \coloneqq \sum_{i=1}^{d} \ketbraa{\pi_\alpha(i)}{i}$ are the unitaries that implement them, and the matrices $D_\alpha\coloneqq \sum_{i=1}^d d_\alpha(i) \ketbra{i}{i}$ are all diagonal. This representation has some technical issues when input and output dimensions are different, but this is irrelevant for the present argument \footnote{See the Supplemental Material, where we provide explicit proofs of some of the results discussed in the main text.}. For two arbitrary indices $1\leq i\neq j\leq d$, we can write
\begin{align*}
\left|\Lambda(\rho)_{ij}\right| &\texteq{(i)} \left| \sum\nolimits_\alpha d_\alpha\left(\pi_\alpha(i)\right) d_\alpha\left(\pi_\alpha(j)\right)^* \rho_{\pi_\alpha(i), \pi_\alpha(j)} \right| \\
&\leq \sum\nolimits_\alpha \left| d_\alpha\left(\pi_\alpha(i)\right)\right| \left|d_\alpha\left(\pi_\alpha(j)\right)\right|\, |\rho_{\pi_\alpha(i), \pi_\alpha(j)}| \\
&\textleq{(ii)} \eta(\rho) \sum\nolimits_\alpha \left| d_\alpha\left(\pi_\alpha(i)\right)\right| \left|d_\alpha\left(\pi_\alpha(j)\right)\right| \\&\quad \times \sqrt{\rho_{\pi_\alpha(i),\pi_\alpha(i)}}\sqrt{\rho_{\pi_\alpha(j), \pi_\alpha(j)}} \\
&\textleq{(iii)} \eta(\rho) \left( \sum\nolimits_\alpha \left| d_\alpha\left(\pi_\alpha(i)\right)\right|^2 \rho_{\pi_\alpha(i),\pi_\alpha(i)}\right)^{1/2} \\&\quad \times \left( \sum\nolimits_\alpha \left| d_\alpha\left(\pi_\alpha(j)\right)\right|^2 \rho_{\pi_\alpha(j), \pi_\alpha(j)}\right)^{1/2} \\
&\texteq{(iv)} \eta(\rho) \mbox{$\sqrt{\Lambda(\rho)_{ii} \Lambda(\rho)_{jj}}$}\, .
\end{align*}
The above steps are justified as follows: (i) we employed the Kraus representation of Eq.~\eqref{SIO Kraus}; (ii) since $i\neq j$ by hypothesis, for all permutations $\pi_\alpha$ we have also $\pi_\alpha(i)\neq \pi_\alpha(j)$ and hence $|\rho_{\pi_\alpha(i), \pi_\alpha(j)}|\leq \eta(\rho) \sqrt{\rho_{\pi_\alpha(i),\pi_\alpha(i)}}\sqrt{\rho_{\pi_\alpha(j), \pi_\alpha(j)}}$; (iii) we applied the Cauchy--Schwarz inequality; (iv) we resorted once more to the representation in Eq.~\eqref{SIO Kraus}.

Another property of $\eta$ which we note is its lower semicontinuity (l.s.c.), that is, the fact that for any sequence $\{\rho_k\}_k$ converging to $\rho$, we have $\eta(\rho)\leq \liminf_{k \to\infty} \eta(\rho_k)$. This follows by noting that $\eta$ can be expressed the maximum over indices $i \neq j$ of the functions $\rho \mapsto f_{ij}(\rho)g_{ij}(\rho)$, where $f_{ij}(\rho) = |\rho_{ij}|$, while $g_{ij}(\rho)$ is defined as $g_{ij}(\rho) = 1\big/\!\!\sqrt{\rho_{ii}\rho_{jj}}$ if both $\rho_{ii},\rho_{jj} \neq 0$, and $g_{ij}(\rho)=0$ otherwise. Both $f_{ij}$ and $g_{ij}$ can be noticed to be nonnegative l.s.c.\ functions, which means that their product will also be l.s.c.; $\eta$ is then a maximum of a finite family of l.s.c.\ functions, and thus is l.s.c.\ itself.

We remark that the measure $\eta$ can be related to similarly defined measures of maximal correlation between classical random variables~\cite{Hirschfeld1935,Gebelein1941,Witsenhausen1975} and quantum states~\cite{Beigi2013,Delgosha2014,Beigi2015,Beigi2018}. In particular, observe that the maximal coherence measure introduced here can in fact be computed as the quantum maximal correlation~\cite{Beigi2013} of the corresponding maximally correlated state $\rho' = \sum_{i,j} \rho_{ij} \ketbra{ii}{jj}$. However, importantly, the monotonicity of the former under SIO does \emph{not} follow from the monotonicity of the latter under local operations, as not every SIO corresponds to a local operation acting on $\rho'$.

\textbf{\em Distillability criterion}. --- %
The task of coherence distillation~\cite{winter_2016,regula_2017,zhao_2018-1} is concerned with the conversion of general quantum states into maximally coherent states $\Psi_m$. The error in the distillation of a state under a set of quantum channels $\O$ is characterized by the fidelity of distillation
\begin{equation}\begin{aligned}
	F_\O (\rho, m) \coloneqq \sup_{\Lambda \in \O} F\left(\Lambda(\rho), \Psi_m\right) ,
\end{aligned}\end{equation}
where $F(\sigma,\omega)\coloneqq \left\|\!\sqrt\sigma \!\sqrt\omega\right\|_1^2$. The (asymptotic) \textit{distillable coherence} is then the maximal rate at which independent and identically distributed copies of a quantum state can be transformed into copies of the maximally coherent qubit state $\Psi_2$ (coherence bit) with asymptotically vanishing error; precisely, we have
\begin{equation}\begin{aligned}
	C_{d,\O} (\rho) \coloneqq \sup \Big\{\:r \ \Big| \  \lim_{n\to\infty} F_\O(\rho^{\ox n}, 2^{rn}) = 1 \:\Big\}.
\end{aligned}\end{equation}
We will say that a state $\rho$ is distillable under $\O$ if $C_{d,\O} (\rho) > 0$.

We now make a crucial observation which lets us immediately relate the maximal coherence $\eta$ to the problem of coherence distillation under SIO. It is the fact that $\eta$ obeys the so-called \textit{tensorization property} \footnote{We note that the existence of a monotone obeying the tensorization property has also been used to study distillation of quantum resources in continuous variable systems~\cite{giedke_2002,lami_2018}.}\nocite{giedke_2002,lami_2018}, that is,
$\eta \left( \rho \otimes \sigma \right) = \max\left\{ \eta(\rho), \,\eta(\sigma) \right\}$, $\forall\ \rho,\sigma $.
To prove this identity, observe that, according to Eq.~\eqref{eta}, computing ${\eta(\rho\otimes\sigma)}$ corresponds to maximizing the function $\big|(\rho\otimes\sigma)_{ik,jl}\big|\,\Big/\!\sqrt{(\rho\otimes \sigma)_{ik,ik} (\rho\otimes\sigma)_{jl,jl}}$ over all pairs of indices $(ik) \neq (jl)$, and that this is equivalent to maximizing $\left(|\rho_{ij}|\big/\!\sqrt{\rho_{ii}\rho_{jj}}\right)\left(|\sigma_{kl}|\big/\!\sqrt{\sigma_{kk}\sigma_{ll}}\vphantom{\sqrt{\rho_{ij}}}\right)$ over choices such that $i \neq j$ or $k \neq l$. The latter maximum is achieved on pairs either of the form $(ik,jk)$ with $i\neq j$ or of the form $(ik,il)$ with $k\neq l$, which corresponds precisely to the larger of $\eta(\rho)$ and $\eta(\sigma)$.

By the tensorization and  monotonicity of $\eta$, we readily obtain one of our main results: a necessary and sufficient criterion for the distillability of an arbitrary quantum state under SIO.
\begin{theorem}\label{thm distillability}
For all states $\rho$, the following are equivalent: (a) $C_{d,\SIO}(\rho) > 0$; (b) $C_{d,\PIO}(\rho)>0$; and (c) $\eta(\rho) = 1$.
\end{theorem}
\begin{proof}
Noting that $\eta(\Psi_2) = 1$ and remembering that $\eta$ is l.s.c., we see that for SIO distillation to be possible, there needs to exist a sequence of SIO operations $\Lambda_n$ such that $\eta\!\left(\lim_{n \to \infty}\Lambda_n(\rho^{\otimes n})\right) = 1$. However, it holds that
\begin{equation*}\begin{aligned}
	\eta\!\left(\lim_{n \to \infty} \!\Lambda_n\!(\rho^{\otimes n})\!\right) \!\textleq{(i)} \liminf_{n\to\infty} \eta\!\left(\Lambda_n(\rho^{\otimes n})\!\right) \!\textleq{(ii)}\! \lim_{n\to\infty} \!\eta\!\left(\rho^{\otimes n}\!\right) \texteq{(iii)} \eta(\rho) ,
\end{aligned}\end{equation*}
where (i) is due to the l.s.c.~of $\eta$, (ii) comes from its monotonicity, and (iii) follows from the tensorization property. This shows in particular that any state with $\eta(\rho) < 1$ is SIO (hence PIO) undistillable. Conversely, the PIO protocol given in the proof of the forthcoming Theorem~\ref{Q lower bound prop} (see below), shows that every state with $\eta(\rho)=1$ is PIO distillable.
\end{proof}
The above Theorem~\ref{thm distillability} establishes a complete characterization of distillability under SIO and PIO. In particular, it is not difficult to see that any generic quantum state exhibits bound coherence, and that the condition for distillability of a state $\rho$ --- i.e., the existence of a submatrix of $\rho$ in the basis $\{\ket{i}\}_i$ proportional to a pure state --- is an extremely restrictive property, satisfied only by a zero-measure class of mixed states. We stress that the proof of the Theorem in fact establishes the stronger statement that any state $\rho$ such that $\eta(\rho)<1$ cannot be used to distill even a single coherence bit, no matter how large the number of available copies of $\rho$ is.
We will see later that this relation between SIO and PIO extends beyond the distillability criterion.

\textbf{\em Fidelity of distillation under SIO}. --- %
It follows from~\cite[Thm.~10]{zhao_2018-1} that the fidelity of distillation of an $m$-dimensional maximally coherent state $\Psi_m$ under SIO for any state $\rho$ is
\begin{equation}
F_\SIO(\rho,m)\! =\! \max\! \bigg\{\! \Tr\! \rho A\! \ \big|\,  0\! \leq\! A\!\! \leq\! \id,\, \Delta( A\!)\! =\!\! \frac{\id}{m},\, \CN(A\!) \!\leq\! m \bigg\} ,
\label{fidelity distillation SIO}
\end{equation}
where the coherence number $\CN(A)$ of $A\geq 0$ is defined as the minimal integer $r$ such that $A$ can be written as a positive linear combination of rank-one projectors $\ketbra{x_i}{x_i}$ with $\rank\left(\Delta(\ketbra{x_i}{x_i})\right) \leq r$ for all $i$~\cite{chin_2017,regula_2018-2}. For the case of distilling a coherence bit $\Psi_2$, we are able to obtain the following simplified characterization.
\begin{theorem}\label{thm:sdp}
The fidelity of distillation of a single bit of coherence under SIO is given by the SDP
\begin{subequations}\label{eq:fidelity_sdp}
\begin{align}
	F_\SIO(\rho,2) &= \max_{\substack{-\id \leq X \leq \id\\\Delta(X)=0\\X \mgeq 0}} \frac{1}{2}\Big(\Tr |\rho| X + 1 \Big)\label{eq:sdp_max}
	\\&= \min_{\substack{D = \Delta(D)\\N \mgeq 0}} \frac{1}{2} \Big(\norm{|\rho| + D + N}{1} + 1 \Big) ,\label{eq:sdp_min}
	\end{align}
\end{subequations}
where $A\succeq 0$ signifies the entrywise inequality $A_{ij}\geq 0$ for all $i,j$, and $|\rho|$ stands for the entrywise modulus of $\rho$.
\end{theorem}

\begin{proof}
We sketch the main idea of the argument, deferring the details to the SM~\cite{Note1}. The expression in Eq.\eqref{fidelity distillation SIO} for the distillation fidelity involves the nontrivial constraint $\CN(A)\leq m$ on the coherence number of the variable $A\geq 0$. For $m=2$, this can be cast into an analytically manageable form thanks to~\cite[Thm.~1]{ringbauer_2017}, which states that $\CN(A)\leq 2$ iff $2\Delta(A) - |A|\geq 0$. By leveraging this criterion and choosing carefully the optimization variables, one arrives at Eq.~\eqref{eq:sdp_max}. Finally, Eq.~\eqref{eq:sdp_min} is obtained by taking the SDP dual.
\end{proof}

The above results can be compared with analogous expressions for $F_\MIO(\rho,2)$ and $F_\DIO(\rho,2)$~\cite{regula_2017}.
In particular, it is known that $F_\MIO(\psi,m) = F_\SIO(\psi,m)$ for all pure states $\psi$ and all $m$~\cite{regula_2017}.
It is left to determine how closely one can {\em approximate} distillation of a perfect bit of coherence by means of SIO when one is given a large number of copies of an input state. This leads us to investigate the quantity $F_{\SIO}(\rho^{\otimes n}, 2)$ as a function of $\rho$ and $n$, and in particular its asymptotic properties in the limit of large $n$. The following result, whose full proof we provide in the SM~\cite{Note1}, provides an operational interpretation of the SIO monotone $\eta$ introduced here.
\begin{theorem} \label{F cosbit thm}
For all states $\rho$ it holds that
\begin{equation}\begin{aligned}
	\lim_{n\to \infty} F_{\SIO}\left(\rho^{\otimes n},2\right) = \frac{1+\eta(\rho)}{2}\, ,
\label{limit F SIO}
\end{aligned}\end{equation}
and the convergence in the above identity is exponentially fast.
\end{theorem}

As a particularly strong example of SIO/PIO undistillability, consider the class of qubit states $\rho_\lambda = \lambda \Psi_2 + (1-\lambda) \frac\id2$ with $\lambda \in [0,1]$. An explicit computation yields $\eta(\rho_\lambda) = \lambda$. By constructing a suitable choice of feasible solutions for the SDP~\eqref{eq:fidelity_sdp}~\cite{Note1}, it can be shown that $F_\SIO(\rho_\lambda^{\otimes n}, 2) = (1+\lambda)/2$ for any number of copies $n$. Therefore, not only is the distillation of $\rho_\lambda$ impossible under SIO for $\lambda\neq 1$, it actually is impossible to increase the fidelity of distillation whatsoever by adding more copies of the state.

\textbf{\em Distillable coherence under SIO and PIO.}---%
Although we have proven that most states are bound coherent under SIO/PIO, it could be nevertheless interesting to compute the amount of coherence $C_{d,\SIO/\PIO}$ that can be extracted from distillable states. This is a very different scenario from that considered in Theorem~\ref{F cosbit thm}: while there we were interested in the distillation of \emph{a single} coherence bit with good fidelity, here we look at the maximal \emph{rate} of distillation of bits of coherence with vanishing errors.

Motivated by the properties of the monotone $\eta(\rho)$, we will now consider a quantifier which we will relate to the distillable coherence. For a state $\rho$ such that $\Delta(\rho)>0$, construct the set $E_\rho\coloneqq\left\{ (i,j):\ |\rho_{ij}|= \sqrt{\rho_{ii}\rho_{jj}} \right\}$. As we show in the SM~\cite{Note1}, it turns out that there is a partition $\{I^\rho_s\}_{s\in\mathcal{S}^\rho}$ of $\{1,\ldots, d\}$ such that $(i,j)\in E_\rho$ iff $i,j$ belong to the same set $I^\rho_s$. With this observation, one can show that the operator $\widebar{\rho} \coloneqq \sum\nolimits_{\substack{(i,j)\in E_\rho}} \rho_{ij} \ketbraa{i}{j}$ is a legitimate density matrix, and that the quantifier
\begin{equation}
    Q(\rho)\coloneqq S\left(\Delta(\rho)\right) - S\left(\widebar{\rho}\right)
\label{Q}
\end{equation}
is: (i) nonnegative; (ii) strictly positive iff $\eta(\rho)=1$; and (iii) additive over tensor products.

We will now show that $Q$ in fact exactly quantifies the SIO and PIO distillable coherence of any state. 
The result will strengthen the relation between these two classes of operations, showing that --- in contrast to task such as coherence dilution, where SIO is as powerful as larger sets of operations, in distillation the power of SIO is actually the same as PIO, where the latter is known to define a very limited framework~\cite{chitambar_2016-1}. We note that coherence distillation under PIO beyond pure states has not been characterized before in any way~\cite{streltsov_2017}. As usual, the process of evaluating a maximal distillation rate is composed of two parts. First, one designs a protocol that achieves the conjectured rate in the limit of a large number of copies (direct part). Second, one shows that the performance of this protocol can not be beaten at least asymptotically (converse part).

\begin{theorem} \label{Q lower bound prop}
For all states $\rho$, the distillable coherence under SIO and PIO satisfies $C_{d,\SIO}(\rho) = C_{d,\mathrm{PIO}}(\rho) = Q(\rho)$.
\end{theorem}

\begin{proof}
To establish that $Q(\rho)$ gives a lower bound to the rate of PIO distillation, given $n$ copies of the state $\rho$, we perform independently on each of them the measurement $\{\Pi_{I^\rho_s} \}_{s\in \mathcal{S}^\rho}$, where $\Pi_{I^\rho_s} \coloneqq \sum_{i\in I^\rho_s} \ketbra{i}{i}$ and $\{I^\rho_s\}_{s \in \mathcal{S}^\rho}$ is the partition of $\{1,\ldots, d\}$ identified above. Setting $P(s)\coloneqq \Tr[\rho\Pi_{I^\rho_s}]$ and $\widebar{\rho}_s\coloneqq P(s)^{-1} \Pi_{I^\rho_s}\rho\Pi_{I^\rho_s}$, we see that this protocol produces an average of $nP(s)$ copies of the states $\widebar{\rho}_s$, which can be shown to be all pure. It is known~\cite{chitambar_2016-1} that there exists a PIO protocol that extracts $S\left( \Delta(\psi)\right)$ coherence bits per copy out of any pure state $\psi$. Applying this procedure to each $\widebar{\rho}_s$ leads to an expected number of coherence bits produced equal to $\sum_{s\in\mathcal{S}^\rho} n P(s) S\left( \Delta(\widebar{\rho}_s)\right) = n Q(\rho)$, achieving a rate $Q(\rho)$. See~\cite{Note1} for further technical details.

To show the converse, consider the family $\{I^\rho_s\}_{s\in\mathcal{S}^\rho}$ of disjoint subsets of $[d]$ as discussed above. For any other state $\sigma$, we can then construct a random variable $S^\rho_\sigma$ supported on $\mathcal{S}^\rho$ whose probability distribution takes the form $P_{S^\rho_\sigma}(s)\coloneqq \Tr\left[\sigma \Pi_{{I_s^\rho}}\right]$. Clearly, $S^\rho_\sigma$ is a coarse-grained version of the random variable $I_\sigma$ distributed according to $P_{I_\sigma} (i)\coloneqq \sigma_{ii} = \braket{i|\sigma|i}$. A first important observation is that the quantifier $Q$ coincides with the conditional entropy of $I_\rho$ given $S^\rho_\rho$: $Q(\rho) = H(I_\rho |S^\rho_\rho)$. 
To establish that $Q$ gives the asymptotic rate of distillation exactly, we will employ the family of monotones defined as \cite{GrandTour}
\begin{equation}\begin{aligned}
	\mu_k(\rho)\coloneqq \max_{I\subseteq [d],\, |I|\leq k} \log \left\|\Pi_I \Delta(\rho)^{-1/2}\rho\,\Delta(\rho)^{-1/2} \Pi_I \right\|_\infty
\end{aligned}\end{equation}
where for $I\subseteq [d]$ we set $\Pi_I\coloneqq \sum_{i\in I} \ketbra{i}{i}$, and the inverse of $\Delta(\rho)$ is taken on the support. These functions can be thought of as a generalization of the previously introduced $\eta$, as $\mu_2(\rho) = \log(1+\eta(\rho))$. The proof proceeds by showing that by suitably smoothing the quantities $\mu_k$, they can be related with a family of smoothed conditional max entropies, which then can be related to $H(I_\rho |S^\rho_\rho)$ by establishing a tweaked asymptotic equipartition property. Using the monotonicity of the family $\mu_k$ under SIO, we can then show that in the limit of infinitely many copies of $\rho$ the achievable rates of distillation under SIO are constrained precisely as $C_{d,\SIO}(\rho) \leq Q(\rho)$. We refer to~\cite{Note1,GrandTour} for the complete technical details of the proof.
\end{proof}

\textbf{\em Conclusions.}--- We fully characterized the problem of asymptotic distillability of quantum coherence under strictly incoherent operations (SIO) and physically incoherent operations (PIO), analytically computing the maximal asymptotic distillation rates and showing that they coincide on all states. We showed that almost all states --- with the sole exception of states whose density matrix contains a rank-one submatrix --- are bound coherent. A new SIO monotone, the maximal coherence $\eta$, plays a crucial role in forming a necessary and sufficient criterion for distillability. We furthermore derived a computable SDP expression for the fidelity of one-shot distillation of a coherence bit under SIO and evaluated it in the asymptotic many-copy limit in terms of the monotone $\eta$.

Our results reveal that, despite being as useful as the larger classes of free operations IO, DIO, and MIO in some tasks, the operational capabilities of SIO and PIO are limited in the context of coherence distillation. This a priori unexpected conclusion was not suggested by any previous work, and bears a notable impact on practical applications, which often require the use of coherence in pure, distilled form~\cite{yuan_2015, streltsov_2016, anshu_2018, diaz_2018}. For those states that happen to be SIO distillable, we constructed a protocol to perform optimal distillation, that should be easily implementable as it requires only incoherent ancillae, incoherent unitaries, and incoherent measurements.

\begin{acknowledgments}
We note the similarity of our main result to Ref.~\cite{marvian_2018}, where a generic phenomenon of bound coherence was also found in the related resource theory of unspeakable coherence (a.k.a.~asymmetry) with respect to the set of translationally-covariant incoherent operations~\cite{marvian_2016}; however, it does not appear possible to make this qualitative correspondence also quantitative, as the two settings are fundamentally different.

In light of the considerations of our work and the exposed weakness of SIO in performing coherence distillation, it remains an important open question to understand what the smallest physically-motivated set of free operations for manipulating coherence without such hindering operational limitations could be, and hence the ongoing quest for a satisfactory resource theory of coherence~\cite{streltsov_2017} becomes even more enthralling.

\textbf{\em Acknowledgments.}--- We are grateful to Iman Marvian, Alexander Streltsov, and Andreas Winter for useful discussions, and to the authors of~\cite{zhao_2018-1} for sharing a preliminary draft of their work with us. We thank the Isaac Newton Institute for Mathematical Sciences for support and hospitality during the programme `Beyond I.I.D. in Information Theory' when part of the work on this paper was undertaken. We acknowledge financial support from the European Research Council (ERC) under the Starting Grant GQCOP (Grant No.~637352).
\end{acknowledgments}

\bibliographystyle{apsrev4-1}
\bibliography{main}

\clearpage
\onecolumngrid
\begin{center}
\vspace*{\baselineskip}
{\textbf{\large Supplemental Material: \\[3pt] Generic bound coherence under strictly incoherent operations}}\\[1pt] \quad \\
\end{center}

\renewcommand{\theequation}{S\arabic{equation}}
\setcounter{equation}{0}
\setcounter{figure}{0}
\setcounter{table}{0}
\setcounter{section}{0}
\setcounter{page}{1}
\makeatletter

\section{SIO and the case of different input and output dimensions}

As reported in the main text, an quantum channel $\Lambda$ acting on a $d$-dimensional system and outputting an $m$-dimensional system is defined to be an SIO if it admits a Kraus representation $\Lambda(\cdot) = \sum_\alpha K_\alpha (\cdot) K_\alpha^\dagger$ such that
\begin{equation}
K_\alpha \Delta(\rho) K_\alpha^\dagger = \Delta(K_\alpha \rho K_\alpha^\dagger) \qquad \forall \ \alpha, \ \forall\ \rho ,
\label{SIO Kraus commute dephasing}
\end{equation}
where $\Delta$ denotes the dephasing map acting on systems of the appropriate dimension. It is easy to verify that when $m=d$ the validity of Eq.~\eqref{SIO Kraus commute dephasing} entails that
\begin{align}
K_\alpha &= U_{\pi_\alpha}^\intercal D_\alpha , \label{Kraus decomposition 11} \\
U_{\pi_\alpha} &= \sum_{i=1}^d \ketbra{\pi_\alpha(i)}{i}, \label{Kraus decomposition 12} \\
D_\alpha &= \sum_{i=1}^d d_\alpha(i) \ketbra{i}{i} , \label{Kraus decomposition 13}
\end{align}
where the $\pi_\alpha$ are permutations. However, it is no longer possible to obtain this simple form when $m\neq d$. To see why, start by observing that each Kraus operator should now become a $m\times d$ rectangular matrix. In the case where $m\geq d$, it is still possible to represent $K_\alpha$ as in Eq.~\eqref{Kraus decomposition 11}, provided that one makes the sum in Eq.~\eqref{Kraus decomposition 12} run all the way to $m$. The opposite case $m\leq d$ can be treated in a similar way by exchanging the order of the product in Eq.~\eqref{Kraus decomposition 11}, which corresponds to applying the above procedure to $\Lambda^\dag$ instead of $\Lambda$.

It is very convenient to have a representation of the Kraus operators of an SIO operation, that is valid for all $d$ and $m$. This can be obtained in two different ways. On the one hand, we can write
\begin{align}
K_\alpha &= U_{\pi_\alpha}^\intercal D_\alpha , \label{Kraus decomposition 31} \\
U_{\pi_\alpha} &= \sum_{i\in J_\alpha} \ketbra{\pi_\alpha(i)}{i}, \label{Kraus decomposition 32}\\
D_\alpha &= \sum_{i\in J_\alpha} d_\alpha(i) \ketbra{i}{i} , \label{Kraus decomposition 33}
\end{align}
where $J_\alpha\subseteq \{1,\ldots, m\}$ are subsets and $\pi_\alpha:J_\alpha\to \{1,\ldots, d\}$ are injective functions. On the other hand, we can resort to the (non-unique) decomposition
\begin{equation}
K_\alpha = U_{\sigma_\alpha}^\intercal D_\alpha U_{\pi_\alpha}^\intercal ,
\label{Kraus decomposition 4}
\end{equation}
where $\pi_\alpha$ and $\sigma_\alpha$ are permutations acting on $\{1,\ldots, d\}$ and $\{1,\ldots, m\}$, respectively.

The careful reader will have noticed that in our proof of the monotonicity of the maximal coherence $\eta$ under SIO, we restricted ourselves to the case where input and output dimensions coincide. This is possible without loss of generality, because of the following ``lift and compress'' argument. Given an SIO channel $\Lambda$ that acts on a $d$-dimensional system and outputs an $m$-dimensional system, and taken some $d'\geq \max\{d,m\}$, construct the modified SIO channel $\Lambda'$ that acts on a $d'$-dimensional system and is defined by the formula
\begin{equation}
\Lambda'(\rho)\coloneqq \Pi_{m}^\intercal \Lambda\left( \Pi_{d} \rho \Pi_{d}^\intercal\right) \Pi_{m} ,
\label{Lambda tilde 1}
\end{equation}
where $\rho$ is $d'\times d'$, and $\Pi_r:\CC^{d'}\to \CC^r$ denotes the projector onto the subspace spanned by the first $r$ basis vector. Observe that from Eq.~\eqref{Lambda tilde 1} we can deduce the identity
\begin{equation}
\Lambda(\rho) = \Pi_{m} \Lambda'\left( \Pi_{d}^\intercal \rho \Pi_{d} \right) \Pi_{m}^\intercal ,
\label{Lambda tilde 2}
\end{equation}
valid for all $d\times d$ matrices $\rho$. Since in the main text we proved that $\eta$ is monotonic at least under strictly incoherent operations that do not change the input dimension, we know that it is monotonic in particular under $\Lambda'$. We now show how to deduce from this that it is also monotonic under $\Lambda$. Evaluating $\eta$ on both sides of Eq.~\eqref{Lambda tilde 2}, and using the elementary observation that $\eta\left( \Pi_r \rho\Pi_r^\intercal\right)\leq \eta(\rho)$, as an inspection of Eq.~\eqref{eta} immediately reveals, we deduce that
\begin{equation*}
\eta\left(\Lambda(\rho)\right) = \eta\left( \Pi_{m} \Lambda'\left( \Pi_{d}^\intercal \rho \Pi_{d} \right) \Pi_{m}^\intercal\right) \leq \eta\left( \Lambda'\left( \Pi_{d}^\intercal \rho \Pi_{d} \right) \right) \leq \eta\left( \Pi_{d}^\intercal \rho \Pi_{d} \right) = \eta(\rho) ,
\end{equation*}
for all $d\times d$ density matrices $\rho$.
This proves that the maximal coherence is in fact monotonic under general SIO operations.

\vspace{5ex}
\section{Fidelity of distillation under SIO}

In what follows, we will denote by $A\circ B$ the Hadamard (or Schur, or entrywise) product of two matrices $A$ and $B$ of the same size. Explicitly, we have that $(A\circ B)_{ij}\coloneqq A_{ij} B_{ij}$.

\begingroup
\renewcommand\thedefinition{\ref{thm:sdp}}
\begin{theorem}
The fidelity of distillation of a single bit of coherence is given by the semidefinite program
\begin{subequations}\label{eq:fidelity_sdp_s}
\begin{align}
	F_\SIO(\rho,2) &= \max_{\substack{-\id \leq X \leq \id\\\Delta(X)=0\\X \mgeq 0}} \frac{1}{2}\Big(\Tr |\rho| X + 1 \Big)\tag{\ref{eq:sdp_max}}
	\\&= \min_{\substack{D = \Delta(D)\\N \mgeq 0}} \frac{1}{2} \Big(\norm{|\rho| + D + N}{1} + 1 \Big).\tag{\ref{eq:sdp_min}}
	\end{align}
\end{subequations}
\end{theorem}
\endgroup

\begin{proof}
We start by recalling the expression for the fidelity of coherence distillation under SIO given in~\cite[Thm.~10]{zhao_2018-1}:
\begin{equation}
F_\SIO(\rho,m) = \max \bigg\{ \Tr[\rho A] \ \big|\,  0 \leq A \leq \id,\, \Delta(A) =\frac{\id}{m},\, \CN(A) \leq m \bigg\} .
\tag{\ref{fidelity distillation SIO}}
\end{equation}
Since we are interested in distilling a single bit of coherence, we set $m=2$ hereafter. Introducing the alternative parametrization $A = \frac{\id + Y}{2}$, Eq.~\eqref{fidelity distillation SIO} can be expressed as the maximization of the function $\frac12(\Tr(\rho Y) + 1)$ subject to the constraints $-\id \leq Y\leq \id$ and $\Delta(Y)=0$. The remaining condition $\CN(A)\leq 2$ can be imposed by means of \cite[Thm.~1]{ringbauer_2017}, which states that a positive semidefinite matrix $A\geq 0$ satisfies $\CN(A)\leq 2$ if and only if the matrix $2\Delta(A) - |A|$ is positive semidefinite. We thus have
\begin{equation}\begin{aligned}
	0 \leq 2\Delta(A) - |A| = \id - \frac{\id + |Y|}{2} = \frac{\id - |Y|}{2}\, ,
\end{aligned}\end{equation}
i.e.\ $|Y|\leq \id$. Now, we want to argue that this latter condition automatically implies that $-\id \leq Y\leq \id$, or equivalently that $\|Y\|_\infty\leq 1$, which makes this constraint superfluous. To see why, write $\|Y\|_\infty \leq \big\| |Y|\big\|_\infty = \lambda_{\max}(|Y|) \leq 1$, where the steps are justified as follows: the first inequality is well-known, and can be explicitly seen to hold by writing
\begin{align*}
\|Y\|_\infty &= \max_{\|v\|_2=1} |v^\dag Y v| = \max_{\|v\|_2=1} \left| \sum\nolimits_{i,j} v_i^* v_j Y_{ij} \right| \\
&\leq \max_{\|v\|_2=1} \sum\nolimits_{i,j} |v_i| |v_j|\, |Y_{ij}|\\
&= \max_{\|w\|_2=1,\, w\succeq 0} w^\dag |Y| w \leq \max_{\|w\|_2=1} w^\dag |Y| w = \big\| |Y|\big\|_\infty\, ;
\end{align*}
the middle equality follows from the Perron--Frobenius theorem, which implies that the spectral radius of every entrywise nonnegative matrix is itself an eigenvalue, which then by the hermiticity of $|Y|$ coincides with its operator norm; finally, the last inequality is a consequence of the assumption that $|Y|\leq \id$.

Putting everything together, we see that $Y$ is only subjected to the two constraints $\Delta(Y)=0$ and $|Y|\leq \id$ (equivalently, $-\id \leq |Y|\leq \id$). We can thus parametrize $Y=X\circ \omega$, where $X\coloneqq |Y|$ satisfies $\Delta(X)=0$, $X\succeq 0$ and $X\leq \id$ (equivalently, $-\id \leq X\leq \id$), while $\omega=\omega^\dag$ is \emph{any} Hermitian matrix composed only of phases (complex numbers of unit modulus). Since the objective function takes the form
\begin{equation}\begin{aligned}
	\Tr[\rho A] = \frac12 \left( 1 + \sum\nolimits_{i,j} \rho_{ij} Y_{ji} \right) = \frac12 \left( 1 + \sum\nolimits_{i,j} \omega_{ji}\, \rho_{ij} X_{ji} \right) ,
\end{aligned}\end{equation}
it is maximized by the choices $\omega_{ji} = e^{i \mathrm{Arg}\, \rho_{ij}} = \frac{\rho_{ij}}{|\rho_{ij}|}$ (and $\omega_{ji}=1$ if $\rho_{ij}=0$), which --- importantly --- identify a Hermitian matrix $\omega = \omega^\dag$. The resulting value of the objective function is
\begin{equation}\begin{aligned}
	\max_{\omega} \frac12 \left( 1 + \sum\nolimits_{i,j} \omega_{ji}\, \rho_{ij} X_{ji} \right) &= \frac12 \left( 1 + \sum\nolimits_{i,j} |\rho_{ij}| X_{ji} \right)\\ &= \frac12 \left( 1 + \Tr \left[|\rho|\, X\right] \right) .
\end{aligned}\end{equation}
The maximization over $X$ subjected to the aforementioned constraints yields the first line in the statement of the Theorem. The second line is then simply the corresponding dual SDP --- the fact that strong duality holds, and thus the two problems have the same optimal value, can be straightforwardly seen by choosing any matrix $N$ with strictly positive entries as a feasible solution to~\eqref{eq:sdp_min} and employing Slater's theorem~\cite{boyd_2004}.
\end{proof}

\begin{rem}
The results of Theorem~\ref{thm:sdp} should be compared with the expressions for the fidelity of distillation associated with the larger sets DIO and MIO, given by~\cite{regula_2017}
\begin{equation}\begin{aligned}
	F_\MIO(\rho,2) = F_\DIO(\rho,2) = \frac{1}{2} \left( \min_{\substack{D = \Delta(D)}} \norm{\rho + D}{1} + 1 \right).
\end{aligned}\end{equation}
As we already mentioned, it is known that for all pure states $\psi$ and all $m$ one has $F_\MIO(\psi,m) = F_\SIO(\psi,m)$.
\end{rem}

\begingroup
\renewcommand\thedefinition{\ref{F cosbit thm}}
\begin{theorem}
For all states $\rho$ and all integers $n$, one has
\begin{equation}\begin{aligned}
	\frac{1+\eta(\rho)}{2} - \frac{\eta(\rho)}{2} \mu_\rho^n \leq F_{\SIO}\left(\rho^{\otimes n},2\right) \leq \frac{1+\eta(\rho)}{2}\, ,
\end{aligned} \label{bounds F SIO n}
\end{equation}
where $0<\mu_\rho<1$ is a number that depends only on $\rho$. Hence
\begin{equation}\begin{aligned}
	\lim_{n\to \infty} F_{\SIO}\left(\rho^{\otimes n},2\right) = \frac{1+\eta(\rho)}{2}\, ,
\end{aligned} \tag{\ref{limit F SIO}}
\end{equation}
and the convergence in the above identity is exponentially fast.
\end{theorem}
\endgroup
\begin{proof}
We start by proving the upper bound in Eq.~\eqref{bounds F SIO n}. Consider an arbitrary SIO operation $\Lambda$ that maps a system of dimension $d^n$ into a single qubit. Because of the monotonicity and tensorisation properties of the $\eta$ function, we can write $\eta' \coloneqq \eta\left(\Lambda(\rho^{\otimes n})\right)\leq \eta(\rho^{\otimes n}) = \eta(\rho)$. Remembering that $\Lambda(\rho^{\otimes n})$ is a qubit state, this means that there are $0<p<1$ and $\varphi\in\RR$ such that
\bbb
\Lambda(\rho^{\otimes n}) \eqqcolon \begin{pmatrix} p & \eta' \sqrt{p (1-p)} e^{i\varphi} \\ \eta' \sqrt{p (1-p)} e^{-i\varphi} & 1-p \end{pmatrix} .
\eee
The fidelity between the above state and a coherence bit reads
\begin{align*}
F \left( \Lambda(\rho^{\otimes n}), \Psi_2\right) &= \Tr\left[ \Lambda(\rho^{\otimes n}) \Psi_2\right] \\
&= \frac12 \left( 1+ 2 \eta' \sqrt{p(1-p)} \cos (\varphi) \right) \\
&\leq \frac12 \left( 1+ \eta' \right) \\
&\leq \frac12 \left( 1+ \eta(\rho)\right) ,
\end{align*}
where for the first inequality we noted that $2\sqrt{p(1-p)}\leq 1$. Taking the supremum over all SIO $\Lambda$ yields the upper bound in Eq.~\eqref{bounds F SIO n}.

The lower bound can be proved by designing a suitable SIO protocol that achieves the prescribed fidelity on $n$ copies. To do this, without loss of generality we are going to assume that for the particular state $\rho$ we are considering: (a) the maximum in Eq.~\eqref{eta} is achieved on the pair $(i,j)=(1,2)$; and (b) $\rho_{12}$ is real. These two assumptions imply that
\bbb
\eta(\rho) = \frac{\rho_{12}}{\sqrt{\rho_{11}\rho_{22}}}\, .
\eee
Now, we construct a suitable ``diagonal filtering'' SIO instrument $\Lambda_{\mathrm{DF}}$ that maps a $d$-dimensional system into a qubit and is defined by the Kraus operators
\begin{align*}
K_0 &\coloneqq \sqrt{\min\{\rho_{11},\rho_{22}\}} \left(\rho_{11}^{-1/2} \proj{1} + \rho_{22}^{-1/2} \proj{2}\right)\, , \\
K_1 &\coloneqq \sqrt{1-\min\left\{1,\frac{\rho_{22}}{\rho_{11}}\right\}} \proj{1}\, , \\
K_2 &\coloneqq \sqrt{1-\min\left\{1,\frac{\rho_{11}}{\rho_{22}}\right\}} \ketbraa{1}{2}\, ,\\
K_\alpha &\coloneqq \ketbraa{1}{\alpha}\quad \text{for $\alpha=3,\ldots,d$.}
\end{align*}
The probability of getting the outcome $\alpha=0$ when applying the instrument $\Lambda_{\mathrm{DF}}$ on $\rho$ is clearly
\bbb
P(0) = \Tr \left[ K_0 \rho K_0^\dag\right] = 2 \min\{\rho_{11}, \rho_{22}\}\, .
\eee
The post-measurement state conditioned on the outcome $\alpha=0$ is then
\bbb
\tilde{\rho}_0 = \frac{K_0\rho K_0^\dag}{P(0)} = \frac{1}{2}\begin{pmatrix} 1 & \eta(\rho) \\ \eta(\rho) & 1 \end{pmatrix} .
\eee
Let us apply the instrument $\Lambda_{\mathrm{DF}}$ separately on each one of the $n$ copies of $\rho$ we have at our disposal. Since
\bbb
F(\tilde{\rho}_0, \Psi_2) = \frac{1+\eta(\rho)}{2}
\eee
matches the upper bound in Eq.~\eqref{bounds F SIO n}, we have achieved maximal distillation fidelity whenever at least one of the $n$ outcomes we obtain is $\alpha=0$. This happens with probability
\bbb
P_{\mathrm{success}} = 1 - \left(1-P(0)\right)^n = 1 - \left(1 - 2 \min\{\rho_{11}, \rho_{22}\}\right)^n = 1 - \mu_\rho^n\, ,
\eee
where $\mu_\rho\coloneqq 2 \min\{\rho_{11},\rho_{22}\}$. If none of the outcomes is $\alpha=0$, then we can simply output the fixed state $\ket{1}$. The average distillation fidelity of this protocol is
\begin{align*}
\widebar{F} = (1 - \mu_\rho^n) \cdot \frac{1+\eta(\rho)}{2} + \mu_\rho^n \cdot \frac12 = \frac{1+\eta(\rho)}{2} - \frac{\eta(\rho)}{2} \mu_\rho^n\, ,
\end{align*}
reproducing the lower bound in Eq.~\eqref{bounds F SIO n}.
\end{proof}

\begin{proposition}
	For the state $\rho_\lambda = \lambda \Psi_2 + (1-\lambda) \frac{\id}{2}$ with $\lambda \in [0,1]$, it holds that $F_\SIO(\rho_\lambda^{\otimes n}, 2) = \frac{\lambda+1}{2}$ for any $n \in \NN$.
\end{proposition}
\begin{proof}
	We obtain the lower bound of $\frac{\lambda+1}{2}$ by considering $n=1$ and simply noting that $F(\rho_\lambda,\Psi_2) = \frac{\lambda+1}{2}$. To show the upper bound, we consider the SDP~\eqref{eq:sdp_min} for $F_\SIO(\rho_\lambda^{\otimes n}, 2)$ and take as feasible solutions the following choices:
	\begin{equation}\begin{aligned}
		D &= - \frac{1-\lambda}{2^n} \id,\\
		N &= \lambda \Psi_2^{\otimes n} + \frac{1-\lambda}{2^n} \id - \rho_\lambda^{\otimes n},
	\end{aligned}\end{equation}
	so that $|\rho_\lambda^{\otimes n}| + D + N = \lambda \Psi_2^{\ox n}$. It remains to verify that $N$ is a valid feasible solution, that is, that all of its coefficients are non-negative. This follows by noticing that the diagonal elements of $N$ are given as
	\begin{equation}\begin{aligned}
		N_{ii} = \lambda \frac{1}{2^n} + (1-\lambda) \frac{1}{2^n} - [\rho_\lambda^{\otimes n}]_{ii} = \frac{1}{2^n} - \frac{1}{2^n} = 0
	\end{aligned}\end{equation}
and similarly the off-diagonal elements are $N_{ij} = \frac{1}{2^n} - [\rho_\lambda^{\ox n}]_{ij}$. Since the off-diagonal elements of $\rho_\lambda^{\ox n}$ are always of the form $\frac{\lambda^m}{2^n}$ for some $1 \leq m \leq n$, and moreover $\lambda \leq 1$, we get $N_{ij} \geq 0$. Hence
\begin{equation}\begin{aligned}
	F_\SIO(\rho^{\ox n},2) \leq \frac{1}{2}\left(\norm{|\rho_\lambda^{\otimes n}|+D+N}{1}+1\right) = \frac{1}{2}\left(\norm{\lambda \Psi_2^{\ox n}}{1}+1\right) = \frac{1}{2}(\lambda+1)
\end{aligned}\end{equation}
	 as required.
\end{proof}

\section{Distillable coherence under SIO and PIO}

We start from a $d$-dimensional state $\rho$, which is as usual assumed to satisfy $\Delta(\rho)>0$ without loss of generality. Let us define the positive matrix
\bb
A_\rho \coloneqq \Delta(\rho)^{-1/2} \rho\, \Delta(\rho)^{-1/2} ,
\label{A rho}
\ee
which satisfies $(A_\rho)_{ii}\equiv 1$ for all $i$. Consider the graph $G_\rho = (V_\rho, E_\rho)$ with vertices $V_\rho \coloneqq \{1,\ldots,d\}$ and edges
\bbb
E_\rho \coloneqq \left\{ (i,j):\ |(A_\rho)_{ij}|=1 \right\} = \left\{ (i,j):\ |\rho_{ij}|= \sqrt{\rho_{ii}\rho_{jj}} \right\} .
\eee
For simplicity, we have included into $E_\rho$ all ``diagonal'' pairs of the form $(i,i)$. If $E_\rho$ contains only these elements, we say that $E_\rho$ is trivial.

At this point, the idea we may have is that in order to distill asymptotically perfect maximally coherent states with SIO, the only coherence inside $\rho$ that matters is that identified by the entries $\rho_{ij}$ corresponding to pairs $(i,j)\in E_\rho$. We could be tempted to construct a ``trimmed'' state $\widebar{\rho}$ by cutting off all other entries, i.e.
\bb
\widebar{\rho} \coloneqq \sum_{ (i,j)\in E_\rho} \rho_{ij} \ketbraa{i}{j} ,
\label{rho bar}
\ee
and to conjecture that this is the only object that matters when computing the distillable coherence under SIO (we show below that Eq.~\eqref{rho bar} defines indeed a legitimate density matrix). This is indeed the case, as Theorem~\ref{Q lower bound prop} in the main text shows. Let us start our discussion with a technical result that clarifies the structure of the graph $G_\rho$.

\begin{lemma} \label{G rho cliques lemma}
The connected components of the graph $G_\rho$ are all cliques (i.e.\ complete subgraphs). Equivalently, there exists a partition $\{I_s\}_{s\in\mathcal{S}}$ of $\{1,\ldots, d\}$ such that
\bb
(i,j)\in E_\rho \quad \Longleftrightarrow\quad \exists\ s\in\mathcal{S}:\ i,j\in I_s\, .
\label{E rho partition}
\ee
Moreover, setting $\Pi_{I_s}\coloneqq \sum_{i\in I_s} \ketbra{i}{i}$, all submatrices $\Pi_{I_s}\rho\Pi_{I_s}$ have rank one.
\end{lemma}

\begin{proof}
Consider the largest clique of $G_\rho$, whose vertices can be taken to be $\{1,\ldots, k\}$ up to permutations. Note that we may have $k=1$. Hence, $|(A_\rho)_{ij}|=1$ for all $1 \leq i,j\leq k$, i.e., the upper left $k\times k$ principal submatrix of $A_\rho$ is entirely composed of phases (i.e., complex number of modulus $1$). Now, it is well known that it is possible for a positive semidefinite matrix $B$ to have all entries of modulus $1$ if and only if there are real numbers $\varphi_p$ such that $B_{pq} = e^{i(\varphi_p - \varphi_q)}$ for all $p,q$, so that $B$ is a multiple of the projector onto a maximally coherent state. This can be proved explicitly as follows. Let $B$ be of size $k$. Since all entries have modulus one, the Hilbert--Schmidt norm of $B$ evaluates to $\sum_{p} \lambda_p^2 = \|B\|_2^2=\sum_{p,q} |B_{p,q}|^2 = k^2$, where $\lambda_p\geq 0$ are the eigenvalues of $B$. Given that $\sum_p \lambda_p = \Tr B = k$, this is possible iff only one of the eigenvalues is nonzero (and equal to $k$). Hence, $B$ is of rank one. In light of these considerations, up to conjugating by a diagonal unitary we will assume that $(A_\rho)_{ij}\equiv 1$ for all $1\leq i,j\leq k$.

Now, let $1\leq i\leq k$ and $k<j\leq d$ be such that $(i,j)\in E_\rho$, or $(A_\rho)_{ij}=e^{i\varphi}$ for some $\varphi\in\RR$. Again, up to diagonal unitaries we can take $\varphi=0$ and hence $(A_\rho)_{ij}=1$. Moreover, up to permutations we can also set $i=1$ and $j=k+1$. The $(k+1)\times (k+1)$ upper left corner of $A_\rho$ reads
\bbb
\Pi_{k+1} A_\rho \Pi_{k+1} \coloneqq \begin{pmatrix} 1 & 1 & \ldots & 1 & 1 \\ 1 & 1 & \ldots & 1 & * \\ \vdots & \vdots & & \vdots & \vdots \\ 1 & 1 & \ldots & 1 & * \\ 1 & * & \ldots & * & 1 \end{pmatrix} .
\eee
Positivity can be imposed e.g.\ by taking the Schur complement with respect to the $(k+1)$-th diagonal element. By doing this, one can show that all unknown entries (marked with $*$ above) must in fact be $1$ if $A_{\rho}$ has to be positive definite. Hence, $(i,k+1)\in E_\rho$ for all $i=1,\ldots ,k$, and $\{1,\ldots, k+1\}$ is a clique. Since this is in contradiction with the requirement that $\{1,\ldots, k \}$ be the largest clique, we conclude that $\{1,\ldots, k\}$ is a connected component of the graph. Continuing in this manner, we can isolate all connected components one by one.

Denote by $\{I_s\}_{s\in\mathcal{S}}$ the corresponding partition of $\{1,\ldots, d\}$. Since we have shown that $\rank \Pi_{I_s} A_\rho\Pi_{I_s} =1$ for all $s$, using the fact that $[\Pi_{I_s}, D]=0$ for all diagonal matrices $D$ we can write
\bbb
\rank \left(\Pi_{I_s} \rho \Pi_{I_s}\right) = \rank  \left(\Pi_{I_s} \Delta(\rho)^{1/2} A_\rho \Delta(\rho)^{1/2}  \Pi_{I_s} \right) = \rank \left(\Delta(\rho)^{1/2} \Pi_{I_s}  A_\rho  \Pi_{I_s}  \Delta(\rho)^{1/2} \right) = \rank \left(\Pi_{I_s}  A_\rho  \Pi_{I_s}\right) ,
\eee
which concludes the proof.
\end{proof}

\begin{rem}
If $\Delta(\rho)$ is not guaranteed to have full support, the second statement of Lemma~\ref{G rho cliques lemma} retains its validity, provided that one intends the sets $\{I_s\}_{s\in\mathcal{S}}$ as a family of disjoint subsets of $[d]$ that do not necessarily form a partition.
\end{rem}

As we show below, Lemma~\ref{G rho cliques lemma} implies that $\widebar{\rho}$ is positive semidefinite and hence a legitimate density matrix. We can thus consider the quantity $Q$ defined as
\bb
Q(\rho)\coloneqq S\left(\Delta(\rho)\right) - S\left(\widebar{\rho}\right) .
\tag{\ref{Q}}
\ee
The following result summarizes the main elementary properties of this object.

\begin{corollary} \label{Q positive cor}
For all states $\rho$ there exists a completely positive, trace-preserving and unital channel $\mathcal{P}_\rho$ such that the ``trimmed'' state $\widebar{\rho}$ of Eq.~\eqref{rho bar} satisfies $\widebar{\rho} = \mathcal{P}_\rho(\rho)$. In particular, $\widebar{\rho}$ is a legitimate density matrix, and moreover
\bb
0\leq Q(\rho)\leq S(\Delta(\rho)) - S(\rho) = C_{d,\IO} (\rho)\, .
\label{Q positive}
\ee
\end{corollary}

\begin{proof}
With the notation of Lemma~\ref{G rho cliques lemma}, it suffices to set
\bbb
\mathcal{P}_\rho(\cdot) \coloneqq \sum_{s\in\mathcal{S}} \Pi_{I_s}(\cdot) \Pi_{I_s} ,
\eee
where $\Pi_{I_s}\coloneqq \sum_{i\in I_s} \proj{i}$. By virtue of Eq.~\eqref{E rho partition}, it is not difficult to check that indeed $\widebar{\rho}=\mathcal{P}_\rho(\rho)$, which shows that $\widebar{\rho}$ is a density matrix. 
As $S\left( \Delta(\sigma)\right)- S(\sigma)\geq 0$ for all density matrices, and $\Delta(\widebar{\rho})=\Delta(\rho)$, this shows immediately that $Q(\rho)\geq 0$. Moreover, since $\mathcal{P}_\rho$ is clearly completely positive, trace-preserving and unital, and unital channels never decrease the entropy, we have that
\bbb
Q(\rho) = S\left(\Delta(\rho)\right) - S\left(\mathcal{P}_\rho(\rho)\right) \leq S\left(\Delta(\rho)\right) - S\left(\rho\right)  .
\eee
The final inequality in Eq.~\eqref{Q positive} follows from~\cite[Thm.6]{winter_2016}.
\end{proof}

It is not difficult to show that a state $\rho$ satisfies $Q(\rho)>0$ if and only if the set of edges $E_\rho$ is nontrivial, i.e.\ if and only if $\eta(\rho)=1$. We now prove that $C_{d,\SIO}(\rho)\geq C_{d,\PIO}(\rho) \geq Q(\rho)$ for all states $\rho$, implying in particular that $\eta(\rho)=1$ ensures the PIO (and hence SIO) distillability of $\rho$. This was the only missing claim in the proof of Theorem~\ref{thm distillability} given in the main text.

\begingroup
\renewcommand\thedefinition{\ref{Q lower bound prop}}
\begin{proposition}
For all states $\rho$ in any dimension, the SIO and PIO distillable coherence satisfies
\bb
C_{d,\SIO}(\rho) \geq C_{d,\PIO}(\rho) \geq Q(\rho)\, .
\label{Q lower bound}
\ee
\end{proposition}
\endgroup

\begin{proof}
As described in the main text, there is a simple PIO protocol that achieves a rate $Q(\rho)$. This is composed of three steps, that we recall below referring to Lemma~\ref{G rho cliques lemma} for notation.
\begin{enumerate}[(i)]
\item One applies the instrument with Kraus operators $\{\Pi_{I_s}\}_{s\in\mathcal{S}}$ on each of the $n$ copies of $\rho$ that are initially available.
\item In the limit of large $n$, each outcome $s$ is obtained an average number of times equal to $nP(s)$, where $P(s) \coloneqq \Tr [\Pi_{I_s}\rho]$.
\item The post-measurement state corresponding to the outcome $s$, denoted by $\widebar{\rho}_s\coloneqq P(s)^{-1} \Pi_{I_s}\rho \Pi_{I_s}$, is pure by Lemma~\ref{G rho cliques lemma}; it is then known~\cite{yuan_2015} that there is a PIO protocol that extracts coherence bits at a rate $S\left(\Delta(\widebar{\rho}_s)\right)$; since we started with $nP(s)$ states, we obtain $n P(s) S\left(\Delta(\widebar{\rho}_s)\right)$ bits of coherence at the output.
\end{enumerate}
The distillation rate associated with this protocol is then
\begin{align*}
r &= \sum_{s\in\mathcal{S}} P(s)\, S\left(\Delta(\widebar{\rho}_s)\right) \\
&= \sum_{s\in\mathcal{S}} P(s)\, S\left(P(s)^{-1} \Pi_{I_s} \Delta(\rho) \Pi_{I_s}\right) \\
&= S(\Delta(\widebar{\rho})) - S(\widebar{\rho}) \\
&= S(\Delta(\rho)) - S(\widebar{\rho}) \\
&= Q(\rho)\, ,
\end{align*}
as claimed.
This intuitive yet sketchy description of the protocol can be complemented with the following rigorous analysis. Fix an $\epsilon,\delta>0$. By the weak law of large numbers, the number of times the outcome $s$ is obtained in step (ii) will be $N_s\geq n (P(s)-\delta)$ for all $s\in\mathcal{S}$ with probability $P_{n,\delta}$ converging to $1$ as $n\to\infty$. For a fixed $s$, the PIO protocol described in~\cite{yuan_2015} is able to extract from $\widebar{\rho}_s^{\,\otimes N_s}$ a number $\left\lfloor N_s \left(S(\Delta(\widebar{\rho}_s)) - \delta\right)\right\rfloor$ of coherence bits with vanishing error. Using the fact that $S(\Delta(\widebar{\rho}_s))\leq \log d$ and that $P(s)\leq 1$, observe that
\bbb
N_s \left(S(\Delta(\widebar{\rho}_s)) - \delta\right) \geq n (P(s)-\delta ) \left(S(\Delta(\widebar{\rho}_s)) - \delta\right) \geq n \left( P(s) S(\Delta(\widebar{\rho}_s)) - (1+\log d) \delta \right) .
\eee
Up to discarding some of the produced coherence bits, we are thus able to convert $\widebar{\rho}_s^{\,\otimes N_s}$ into a state $\sigma_s$ such that $\left\|\sigma_s - \sigma_s^{\text{ideal}} \right\|_1\leq \epsilon$, where $\sigma_s^{\text{ideal}} \coloneqq \Psi_2^{\otimes n \left( P(s) S(\Delta(\widebar{\rho}_s)) - (1+\log d) \delta \right)}$. With probability $P_{n,\delta}$ approaching one, the outcome of the protocol will then be $\bigotimes_{s\in\mathcal{S}} \sigma_s$. Using a standard telescopic technique together with the fact that the index $s$ can take at most $d$ values, it is not difficult to verify that
\bbb
\left\| \bigotimes_{s\in\mathcal{S}} \sigma_s - \bigotimes_{s\in\mathcal{S}} \sigma_s^{\text{ideal}}\right\|_1 \leq \sum_{s\in\mathcal{S}} \left\| \sigma_s - \sigma_s^{\text{ideal}}\right\|_1 \leq d\epsilon\, .
\eee
Moreover,
\bbb
\bigotimes_{s\in\mathcal{S}} \sigma_s^{\text{ideal}} = \Psi_2^{\otimes n \left( \sum_{s\in\mathcal{S}} P(s) S(\Delta(\widebar{\rho}_s)) - (1+\log d) \delta\right) } = \Psi_2^{\otimes n \left( Q(\rho) - (1+\log d) \delta\right)} \, .
\eee
Since $d$ is fixed and $\epsilon$ is arbitrary, we conclude that $Q(\rho) - (1+\log d)\delta$ is an achievable rate for all $\delta>0$. Taking the supremum in $\delta$, this shows that $C_{d,\PIO}(\rho)\geq Q(\rho)$.
\end{proof}

One could speculate that a better lower bound on $C_{d,\PIO}$ can be obtained by applying the above distillation protocol to many copies of the state $\rho$ simultaneously, which leads to the bound
\bbb
C_{d,\PIO}(\rho)\geq \lim_{n\to\infty} \frac1n Q\left(\rho^{\otimes n}\right) .
\eee
However, it turns out that $Q$ is additive over tensor products, and thus the r.h.s.\ of the above equation coincides with $Q(\rho)$ itself. We conclude our discussion by proving this last property.

\begin{lemma} \label{additivity Q lemma}
For all states $\rho$ and $\sigma$ of any dimension,
\bb
\widebar{\rho\otimes\sigma} = \widebar{\rho}\otimes\,\widebar{\!\sigma}
\label{additivity bar}
\ee
and consequently
\bb
Q\left(\rho\otimes \sigma\right) = Q(\rho)+Q(\sigma)\, .
\label{additivity Q}
\ee
\end{lemma}

\begin{proof}
Observe that our usual assumption that $\Delta(\rho\otimes \sigma)>0$ is equivalent to requiring that both $\Delta(\rho)>0$ and $\Delta(\sigma)>0$. Call $d_1$ the dimension of the space on which $\rho$ acts, and $d_2$ that of the space on which $\sigma$ acts. For all $i,k\in \{1,\ldots, d_1\}$ and $j,l\in \{1,\ldots, d_2\}$, one has
\bbb
\left| \left( \rho\otimes\sigma\right)_{ij,kl}\right| = |\rho_{ik}|\, |\sigma_{jl}| \leq \sqrt{\rho_{ii}\rho_{kk}} \sqrt{\sigma_{jj}\sigma_{ll}} = \sqrt{\left( \rho\otimes\sigma\right)_{ij,ij} \left( \rho\otimes\sigma\right)_{kl,kl}}\, ,
\eee
Hence, the following facts are easily seen to be equivalent: (i) $(ij,kl)\in E_{\rho\otimes \sigma}$; (ii) the above inequality is saturated; (iii) $|\rho_{ik}|=\sqrt{\rho_{ii}\rho_{kk}}$ and $|\sigma_{jl}|=\sqrt{\sigma_{jj}\sigma_{ll}}$; (iv) $(i,k)\in E_\rho$ and $(j,l)\in E_\sigma$. Hence,
\begin{align*}
\widebar{\rho\otimes\sigma} &= \sum_{(ij,kl)\in E_{\rho\otimes \sigma}} \left( \rho\otimes\sigma\right)_{ij,kl} \ketbra{ij}{kl} \\
&= \sum_{(i,k)\in E_{\rho},\, (j,l)\in E_\sigma} \rho_{ik} \sigma_{jl} \ketbra{i}{k}\otimes \ketbra{j}{l} \\
&= \left( \sum_{(i,k)\in E_{\rho}} \rho_{ik} \ketbra{i}{k} \right)\otimes \left( \sum_{(j,l)\in E_\sigma} \sigma_{jl} \ketbra{j}{l} \right) \\
&= \widebar{\rho}\otimes\,\widebar{\!\sigma}
\end{align*}
\end{proof}

\end{document}